\documentclass[english,11pt,a4paper]{article}
\usepackage[numbers,sort&compress]{natbib}
\usepackage{xcolor,graphicx}
\usepackage[bookmarksnumbered,linktocpage,hypertexnames=false,colorlinks=true,linkcolor=blue,urlcolor=blue,citecolor=blue,anchorcolor=green,breaklinks=true,pdfusetitle]{hyperref}
\usepackage{amsfonts,amssymb,amsmath,amsthm}
\usepackage{microtype,mathrsfs,bm,bbm}
\usepackage{authblk,fullpage,mleftright,mathtools,thmtools}
\usepackage{subcaption}
\usepackage[english]{babel}
\usepackage[capitalize,nameinlink]{cleveref}

\newtheorem{thm}{Theorem}\crefname{thm}{Theorem}{Theorems}
\newtheorem{lem}[thm]{Lemma}\crefname{lem}{Lemma}{Lemmas}
\crefname{prop}{Proposition}{Propositions}
\newtheorem{cor}[thm]{Corollary}\crefname{cor}{Corollary}{Corollaries}

\crefname{obs}{Observation}{Observations}
\crefname{figure}{Figure}{Figures}
\newtheorem*{thm*}{Theorem}
\theoremstyle{definition}
\newtheorem{dfn}[thm]{Definition}\crefname{def}{Definition}{Definitions}
\crefname{exa}{Example}{Examples}
\crefname{rmk}{Remark}{Remarks}

\DeclareMathOperator{\tr}{tr}
\DeclareMathOperator{\dist}{dist}
\DeclareMathOperator{\poly}{poly}

\DeclareMathOperator{\CNOT}{CNOT}

\DeclareMathOperator{\easy}{easy}
\DeclareMathOperator{\hard}{hard}
\DeclareMathOperator{\thresh}{thresh}
\DeclareMathOperator{\QEC}{EC}
\DeclareMathOperator{\bad}{bad}
\DeclareMathOperator{\good}{good}

\DeclareMathOperator{\out}{out}

\newcommand{\id}{\mathbbm{1}}
\newcommand{\CC}{\mathbb{C}}

\renewcommand{\O}{\mathcal{O}}
\renewcommand{\H}{\mathcal H}
\newcommand{\V}{\mathcal V}
\newcommand{\eps}{\varepsilon}

\newcommand{\N}{N}

\newcommand{\circuit}{\mathcal C}
\newcommand{\enc}{\mathcal E}
\newcommand{\dagg}{^\dagger}

\newcommand{\ket}[1]{|#1\rangle}
\newcommand{\bra}[1]{\langle #1|}
\newcommand{\proj}[1]{\ket{#1}\!\bra{#1}}
\newcommand{\braket}[1]{\langle #1\rangle}

\newcommand{\postBQP}{\textsf{postBQP}}
\newcommand{\PP}{\textsf{PP}}
\newcommand{\BQP}{\textsf{BQP}}
\newcommand{\nlev}{\textsc{Nev}}
\newcommand{\normalization}{\textsc{Norm}}
\newcommand{\sharpp}{\# \textsf{P}}

\DeclarePairedDelimiter{\norm}{\lVert}{\rVert}
\DeclarePairedDelimiter{\abs}{\lvert}{\rvert}

\usepackage{tikz}
\usetikzlibrary{decorations.markings, decorations.pathreplacing, decorations.pathmorphing, shapes, calc}
\tikzset{
	enc/.style={isosceles triangle,isosceles triangle apex angle=60,draw,fill=white,minimum size = 3mm, rotate=180},
	dec/.style={isosceles triangle,isosceles triangle apex angle=60,draw,fill=white,minimum size = 3mm, rotate=000},
}
\tikzset{
  qtrash/.pic={
    \draw (-0.4,0.8) -- (-0.3,0) -- (0.3,0) -- (0.4,0.8) -- cycle;
    
    \draw (-0.20,0.7) -- (-0.15,0.1);
    \draw (0,0.75) -- (0,0.05);
    \draw (0.20,0.7) -- (0.15,0.1);
    
    \begin{scope}[shift={(0,.8)}]
    \begin{scope}[rotate=-30]
      \draw (-0.45,0.30) rectangle (0.45,0.45); 
      \draw (-0.1,0.45) -- (-0.1,0.55) -- (0.1,0.55) -- (0.1,0.45);
    \end{scope}
    \end{scope}
      }
}

\newcommand{\measureBox}[1]{
	\begin{scope}[shift={(#1)}]
		\filldraw[fill=white] (-0.7,-0.6) rectangle (0.7,0.6);
		\draw (0.5,-.3) arc[start angle=0,end angle=180,radius=0.5];
		\draw[-latex] (0,-.2) -- (0.55,0.45);
	\end{scope}
}
\newcommand{\projectBox}[1]{
	\begin{scope}[shift={(#1)}]
		\filldraw[fill=white] (-0.7,-0.6) rectangle (0.7,0.6);
		\draw (0,0) node {$\proj{0}$};
	\end{scope}
}
\newcommand{\projectOneBox}[1]{
	\begin{scope}[shift={(#1)}]
		\filldraw[fill=white] (-0.7,-0.6) rectangle (0.7,0.6);
		\draw (0,0) node {$\proj{1}$};
	\end{scope}
}
\newcommand{\cnot}[2]{
	\fill (#1) circle (4pt);

	\begin{scope}[shift={(#2)}]
		\draw (0,0) circle (8pt);
		\draw (0,-8pt) -- (0,8pt); 
	\end{scope}

	\draw (#1) -- (#2);
}
\newcommand{\unencodedCircuit}{
	\draw[decorate,decoration={brace,amplitude=6pt}] (-.5,-.3)--(-.5,3.3) node[midway,xshift=-1.2cm,align=center,text width=1.5cm]{$n$ qubits};
	\draw (0, .6) node {$\vdots$};
	\draw (2.5, .6) node {$\vdots$};
	\draw (0, 0) node {$\ket0$};
	\draw (0, 1) node {$\ket0$};
	\draw (0, 2) node {$\ket0$};
	\draw (0, 3) node {$\ket0$};
	\draw (.5, 0)--(3,0);
	\draw (.5, 1)--(3,1);
	\draw (.5, 2)--(4,2);
	\draw (.5, 3)--(4,3);
	\node[draw,fill=white,minimum width = 1cm, minimum height=4cm,align=center] at (1.5,1.5) {$U$};
	\begin{scope}[scale=.7]
		\projectBox{5,4.3}
		\measureBox{5,2.85}
	\end{scope}
	\draw[dashed] (3,0)--(3.8,-.3);
	\draw[dashed] (3,1)--(3.7,-.4);
	\pic at (3.8,-1.3) {qtrash};
}
\newcommand{\encodedCircuit}{
	\begin{scope}[scale=.7]
		\node[fill=blue!20,rounded corners,minimum width=1cm,minimum height=8cm]  at (2,5) {};
		\draw (2,12) node [text width=3cm, align=center]{Encoding to $k$-concatenated QECC};
		\node[fill=blue!20,rounded corners,minimum width=4cm,minimum height=8cm]  at (6,5) {};
		\draw (6,-1.5) node [text width=6cm, align=center]{FT computation; poly$(n,m)$exp$(\O(k))$ gates};
		\node[fill=blue!20,rounded corners,minimum width=2.5cm,minimum height=5.5cm]  at (11.25,6.82) {};
		\draw (11.25,12) node [text width=4.5cm, align=center]{Decoding, projection, measurement};
		\draw[decorate,decoration={brace,amplitude=6pt}] (-.5,-.3)--(-.5,6.3) node[midway,xshift=-1.2cm,align=center,text width=1.5cm]{$n$ logical qubits};
		\draw[decorate,decoration={brace,amplitude=6pt}] (-.5,8-.3)--(-.5,10.3) node[midway,xshift=-1.2cm,align=center,text width=2cm]{$m-1$ additional projection qubits};
		\draw (0, 1.1) node {$\vdots$};
		\draw (2, 1.1) node {$\vdots$};
		\draw (7, 1.1) node {$\vdots$};
		\draw (0, 9.1) node {$\vdots$};
		\draw (2, 9.1) node {$\vdots$};
        \foreach \y in {0, 1, ..., 2}{
            \begin{scope}[shift={(0, 4*\y)}]
                \draw (0, 0) node {$\ket0$};
                \draw (0, 2) node {$\ket0$};
    
                \draw (.5, 0)--(2,0);
                \draw (.5, 2)--(2,2);
                \draw (1.9,-.05) -- (10.1,-.05);
                \draw (1.9,+.05) -- (10.1,+.05);
                \draw (1.9,1.95) -- (10.1,1.95);
                \draw (1.9,2.05) -- (10.1,2.05);
    
                \node[enc] at (2,0){};
                \node[enc] at (2,2){};
            \end{scope}
            \filldraw[fill=white] (4,-.5) rectangle (6,6.5);
            \filldraw[fill=white] (4.1,-.4) rectangle (5.9,6.4);
            \draw (5,3) node {$U_{\mathrm{enc}}$};
        }
        \foreach \y in {1, 2, ..., 2}{
            \begin{scope}[shift={(10, 4*\y)}]
                \draw (0,0)--(2,0);
                \draw (0,2)--(2,2);
                \node[dec] at (0,0){};
                \node[dec] at (0,2){};
                \draw (0, 1.1) node {$\vdots$};
                \projectBox{2,0}
                \projectBox{2,2}
            \end{scope}
        }
        \measureBox{12,4}
        \cnot{7,6}{7,8}
        \cnot{8,6}{8,10}
        \node[rotate=70] at (7.5,9) {$\cdots$};
        		\draw (10.1,-.05)--(11,-.5);
		\draw (10.1,+.05)--(11,-.4);
		\draw[dashed] (11,-.5)--(12,-1);
		\draw[dashed] (11,-.4)--(12,-.9);
		\draw (10.1,2+.05)--(10.15,2.05)--(11.05,.5);
		\draw (10.1,2-.05)--(11,0.4);
		\draw[dashed] (11,0.4)--(11.8,-1);
		\draw[dashed] (11.05,.5)--(11.83,-.9);
		\pic at (12,-2.2) {qtrash};
	\end{scope}
}

\newcommand{\embeddingFig}{
	\begin{tikzpicture}[scale=1,thick,
			tensor/.style={draw,fill=blue!20,minimum width=.5cm,minimum height=1cm},
			midarrow/.style={
					postaction={decorate},
					decoration={markings, mark=at position 0.7 with {\arrow{latex}}}}
		]

		\node at (-1,.5) {\textbf{(a)}};
		\draw (-.5,.3)--(.5,.3);
		\draw (-.5,-.3)--(.5,-.3);
		\node[tensor]  at (0,0) {$U$};
		\node[text width=3cm, align=center] at (0,-1.5) {2-qubit gate};

		\begin{scope}[shift={(3,0)}]
			\node at (-1,.5) {\textbf{(b)}};
			\draw (-.5,.3)--(1.5,.3);
			\draw (-.5,-.3)--(1.5,-.3);
			\node[tensor]  at (0,0) {$U$};
			\node[tensor,fill=red!10]  at (1,0) {$\eta$};
			\node[text width=4cm,align=center] at (.5,-1.5) {gate followed by depolarizing noise};
		\end{scope}
		\begin{scope}[shift={(7,0)}]
			\node at (-1,.5) {\textbf{(c)}};
			\draw (-.5,.3)--(.5,.3);
			\draw (-.5,-.3)--(.5,-.3);
			\draw[color=red] (0,-.8)--(.5,-.8);
			\node[tensor,minimum height=1.5cm]  at (0,-.25) {$V$};
			\node[text width=3cm,align=center] at (0,-1.5) {Stinespring extension};
			\node[text width=3cm,align=center] at (1.25,0) {$=$};
			\begin{scope}[shift={(2.5,0)}]
				\draw (-.8,0)--(.8,0);
				\draw (0,-.8)--(0,.8);
				\draw[color=red] (0,0)--(.5,-.5);
				\node[circle,draw,fill=blue!20,minimum size=20pt,inner sep=0pt] at (0,0) {$T$};
				\node[text width=3cm,align=center] at (0,-1.5) {PEPS tensor};
			\end{scope}
		\end{scope}
    \end{tikzpicture}
}

\newcommand{\CTCFig}{
	\begin{tikzpicture}[scale=1,thick,
			tensor/.style={draw,fill=blue!20,minimum width=.4cm,minimum height=.4cm},
			midarrow/.style={
					postaction={decorate},
					decoration={markings, mark=at position 0.7 with {\arrow{latex}}}}
		]
		\small
		\begin{scope}[scale=.65]
			\node at (0,0) {\textbf{(a)}};
			\draw (0, -3)--(6, -3);
			\draw (2, -2)--(4, -2)--(5,-1)--(4,0)--(2,0)--(1,-1)--(2,-2);
			\node at (3,-.5) {CTC};
			\cnot{3,-3}{3,-2}
			\draw (7.5,-3)--(10.5,-3);
			\node at (6.75,-3) {$=$};
			\projectBox{9,-3}
		\end{scope}
		\begin{scope}[scale=.65,shift={(13,0)}]
			\node at (0,0) {\textbf{(b)}};
			\draw (0, -3)--(6, -3);
			\draw (2, -2)--(4, -2)--(5,-1)--(4,0)--(2,0)--(1,-1)--(2,-2);
			\node at (3,-.5) {CTC};
			\cnot{3.5,-3}{3.5,-2}
			\draw (7.5,-3)--(10.5,-3);
			\node[tensor,fill=red!10]  at (2.5,-2) {$X$};
			\node at (6.75,-3) {$=$};
			\projectOneBox{9,-3}
		\end{scope}
    \end{tikzpicture}
}

\begin{document}
\title{Computational complexity of injective projected entangled pair states}

\author[1]{Dylan Harley}
\author[2]{Freek Witteveen}
\author[1]{Daniel Malz}
\affil[1]{Department of Mathematical Sciences, University of Copenhagen, Denmark}
\affil[2]{Centrum Wiskunde en Informatica and QuSoft, Amsterdam, the Netherlands}

\date{\today}
\maketitle

\begin{abstract}
	Projected entangled pair states (PEPS) constitute a variational family of quantum states with area-law entanglement.
	PEPS are particularly relevant and successful for studying ground states of spatially local Hamiltonians.
	However, computing local expectation values in these states is known to be \postBQP-hard.
	Injective PEPS, where all constituent tensors fulfil an injectivity constraint, are generally believed to be better behaved, because they are unique ground states of spatially local Hamiltonians.
	In this work, we therefore examine how the computational hardness of contraction depends on the injectivity.
	We establish that below a constant positive injectivity threshold, evaluating local observables remains \postBQP-complete, while above a different constant nontrivial threshold there exists an efficient classical algorithm for the task, resolving an open question from (Anshu et al., STOC `24). We do this by proving that noisy postselected quantum computation can be made fault-tolerant. 
\end{abstract}

\section{Introduction}
Tensor network states provide a way of parametrising multipartite quantum states with constrained entanglement. This is especially relevant for ground states of lattice Hamiltonians, which in many cases are believed to exhibit area-law entanglement. Projected entangled pair states (PEPS) are a general class of such tensor network states.
They are defined through a choice of graph (connecting physical sites through virtual bond spaces), and a choice of tensor at each vertex that maps the incoming bond spaces to a local physical space.
PEPS are commonly used to study ground states of lattice Hamiltonians, and correspondingly one often chooses the graph to be a subset of a finite-dimensional lattice, in which case the PEPS have entanglement entropy satisfying an area law.
In this work, we are principally concerned with two-dimensional PEPS.

On a one-dimensional lattice, the PEPS construction gives rise to a class known as matrix product states (MPS).
Generally, computations involving MPS are much more tractable than their extensions to general PEPS; MPS possess a canonical form for ease of classification, and there are rigorous guarantees for the approximation of gapped ground states in one dimension by MPS.
In contrast, the structure theory of PEPS and its canonical forms in two and more dimensions have weaker properties \cite{Cirac2021,Acuaviva2023}, and the precise relation between gapped ground states, area law states and PEPS is not fully understood \cite{Cirac2019}.
Contraction and computation of local observables is efficient for MPS, whereas for PEPS these tasks are known to be \postBQP-complete~\cite{Schuch2007} (roughly equivalent to simulating a quantum computer with access to postselection), and hence beyond the reach of efficient classical or quantum algorithms.
In fact, PEPS can encode undecidable problems, such as equivalence testing (determining whether two different tensors give the same state for all system sizes)~\cite{Scarpa2020}.

In spite of these difficulties, PEPS are of great theoretical and numerical interest.
Many paradigmatic strongly correlated many-body states have exact PEPS representations, and PEPS are a useful tool for manipulating such states and classifying (symmetry-protected) topological order, see Ref.~\cite{Cirac2021} for an overview. On the numerical side, variational optimisation over PEPS is one of the most accurate methods for the simulation of strongly correlated two-dimensional quantum systems in practice~\cite{Schuch2012,Xie2014,Zheng2017}.
It is therefore desirable to obtain a better understanding of the fundamental computational limits of PEPS, especially under natural conditions that make the family of states more ``physically reasonable.''

The most important such constraint is injectivity.
A PEPS tensor $T_x$ is called injective if it is injective as a map from the bond Hilbert spaces to the physical Hilbert space.
We quantify this using the injectivity parameter $\delta$ which is defined as the inverse of the condition number of the tensor.
If all constituent tensors are injective, this leads to desirable properties of the PEPS, all related to a reduction in complexity. In particular,
\begin{enumerate}
	\item testing equivalence becomes decidable and determined by gauge symmetries~\cite{Molnar2018},
	\item the state has a parent Hamiltonian: a local nearest-neighbour Hamiltonian that has the PEPS as its unique ground state. For strongly injective PEPS (made precise below) this parent Hamiltonian is gapped,
	\item the tensors have bounded condition number, so one can expect better numerical stability.
\end{enumerate}

Injectivity is thus related to the computational complexity of PEPS, and in this work we make this connection more precise and robust.
We do so by showing that both regimes (hardness at zero injectivity and easiness at maximal injectivity) extend to a constant (system size independent) range of the injectivity parameter.
The main result of this work is the following.
\begin{thm*}[Informal]
	Consider the task of estimating local observables of a two-dimensional $\delta$-injective PEPS with constant bond dimension to precision $\eps$. If $\delta > \delta_{\easy}$, there is a $\poly(\eps^{-1})$ algorithm, independent of system size. If $\delta < \delta_{\hard}$, computing local observables to constant precision is \postBQP-hard.
\end{thm*}

To establish the easy regime, we use the result that close to maximal injectivity PEPS are the ground states of a gapped Hamiltonian~\cite{Schuch2011}. This implies exponentially decaying correlations, which means that when we calculate the expectation value of a local observable, we can restrict the calculation to a small patch around the observable~\cite{Schwarz2017}.
Hardness of computing local observables at constantly small injectivity is based on a construction similar to Refs \cite{Anshu2024,Malz2025}, where it was shown that one can embed quantum circuits in injective PEPS, but that injectivity introduces noise. This noise, if sufficiently weak, can however be mitigated using error correction.
To show that contracting weakly injective PEPS is still \postBQP-hard, our technical contribution here is to show that one can also make post-selected quantum computation fault tolerant, a result that may be of independent interest.

\subsection*{Relation to previous work}

The \postBQP-completeness (or equivalently \PP-completeness~\cite{Aaronson2005}) of contracting two-dimensional PEPS was shown in Ref.~\cite{Schuch2007} by embedding quantum circuits and postselection operations in them. One can do so in two spatial dimensions by using one dimension for space and one dimension for time.
This construction gives rise to a somewhat pathological non-injective PEPS whose physical degrees of freedom in the bulk are completely decoupled from the bond system (and in a product state).

Our work is inspired by Ref.~\cite{Anshu2024, Malz2025}, where it was shown that circuits can also be embedded in injective PEPS, albeit at the expense of introducing noise into the circuit. As is shown there, computing local observables in injective PEPS below a constant injectivity threshold is \BQP-hard, because the noise can be mitigated using error correction.
In \cite{Malz2025}, it is shown that for isometric tensor network states, which are PEPS in which the tensors can be interpreted as iteratively applying isometries, the problem is \BQP-complete. In this case, the state is by construction normalized, whereas for general PEPS with constant injectivity the normalization may be exponentially small.

For the easy regime, our argument closely follows Ref.~\cite{Schwarz2017}, which gives a \emph{quasipolynomial} algorithm for approximating local observables for PEPS with uniformly gapped parent Hamiltonians, which we adapt to a polynomial time algorithm. That strongly injective PEPS are uniformly gapped follows from the results in Ref.~\cite{Schuch2011}.

Another relevant work is Ref.~\cite{Haferkamp2020}, which studies \emph{average-case} hardness of PEPS contraction. This also concerns injective PEPS (as random PEPS with appropriate physical and bond dimension are almost surely injective). There, the conclusion is that PEPS contraction is $\sharpp$-hard on average. The argument, however, requires the bond dimension to increase with system size (which implies that the injectivity decreases with increasing system size), while for the worst-case hardness in this work bond dimension $D=2$ and constant injectivity suffice to prove hardness.

\subsection*{Discussion and future directions}

\textbf{Normalization, observables and translation invariance.} Our results apply to computing the expectation values of local observables to additive error with respect to the normalized PEPS (\nlev).
A more basic problem is to just compute the normalization of the PEPS (\normalization), but the relation to $\nlev$ is unclear.
For general PEPS, if both $\nlev$ and $\normalization$ are to multiplicative precision, we have $\nlev=\normalization$~\cite{Schuch2007}.
However, if we ask for additive precision, the reductions no longer work.

Another interesting case are isoTNS~\cite{Zaletel2020}, which have unity norm by construction, but \nlev{} is \BQP-complete~\cite{Malz2025}, even when assuming tensors with constant injectivity.
The most interesting case (to us) that our work leaves open are translation-invariant PEPS, which have the same tensor everywhere.
We observe in \cref{sec:hardness uniform} that computing the normalization to exponential additive precision in that case is hard for constant injectivity, but this does not obviously imply that estimating local observables in the normalized state is hard.

\textbf{Learning PEPS.}
Suppose one is given many copies of a quantum state and a graph with the promise that it has a description as a PEPS on that graph with some given bond dimension. How hard is it to learn (not necessarily uniquely determined) tensors that approximate the state well? For MPS this question is well-studied \cite{Cramer2010,Zhao2024,Fanizza2023}, for PEPS less so.
There may be a dichotomy similar to our main result: at sufficiently high injectivity, PEPS are easy to learn, while for small injectivity they are hard to learn (under computational assumptions).

\textbf{Other complexity measures for PEPS contraction.} We have identified regimes that are classically easy and quantumly hard, and of course it is interesting to better understand what is found in between. One direction is further characterization of when quantum computers can efficiently perform tensor network calculations~\cite{Arad2010,Lin2022}. A different interesting direction is opened by Refs~\cite{Chen2025,Jiang2024}. They show that if the tensors making up a PEPS have a bias towards positive numbers, this can give a reduced average-case hardness, pointing to a different complexity transition than the one studied in this work.

\textbf{Hamiltonian complexity.} The motivation of Ref.~\cite{Anshu2024} for the fault-tolerance construction for injective PEPS was to study Hamiltonian complexity theory, and to give an alternative circuit-to-Hamiltonian reduction. We strengthen this connection further by also allowing postselection in the circuit. This may be a useful new tool, especially when one wants to compute ground state energies to high precision \cite{Deshpande2022}.
A direct consequence of our construction is \postBQP-hardness of computing local observables to constant precision of the (unique) ground states of a frustration-free Hamiltonian (since every PEPS with non-zero injectivity has a parent Hamiltonian).

\textbf{Postselected quantum computation.} Our strategy for proving the hardness of contracting weakly injective PEPS involves designing a scheme for noise-robust quantum circuits with noisy postselection \cite{Aaronson2005} which are still able to perform $\postBQP$-complete computations (see \cref{sec:postBQP-proof}). This may be viewed as an extension of the threshold theorem for $\BQP$ circuits \cite{Aharonov1999}, and may be of independent interest; the possibility of making $\postBQP$ noise-robust has previously been raised in the context of complexity theory \cite{aaronson2004quantum,Aaronson2005,abrams1998nonlinear} and in physically-motivated computational models with closed timelike curves \cite{leung2013computational}.

\subsection*{Organization of this work}
We set up notation and the relevant definitions for tensor network states in \cref{sec:injective tensor networks}. We prove efficient evaluation of local observables for low injectivity in \cref{sec:easy regime}. Next, in \cref{sec:postBQP-proof}, we describe an embedding of noisy circuits with postselection into injective PEPS and use this in combination with a procedure for fault-tolerant postselection to prove $\postBQP$-hardness for low injectivity.

\section{Injective projected entangled pair states}\label{sec:injective tensor networks}
We start by formally defining the notion of (injective) PEPS.

We let $G = (W,E)$ be a graph with vertices $W$ and edges $E$.
Each edge $e \in E$ has an associated \emph{bond dimension} $D_e$.
If $e = (xy)$ for $x, y \in W$, we associate Hilbert spaces $\H_{e,x} = \CC^{D_e} = \H_{e,y}$ to the edge.
We now define the following state
\begin{align}\label{eq:link state}
	\ket{\Phi} = \bigotimes_{e \in E} \ket{\phi_e} \qquad \text{where} \quad \ket{\phi_e} = \frac{1}{\sqrt{D_e}} \sum_{i=1}^{D_e} \ket{ii} \in \H_{e,x} \otimes \H_{e,y}.
\end{align}
This state consists of maximally entangled states, distributed according to the graph $G$.

\begin{dfn}[PEPS]\label{def:peps}
	A \emph{projected entangled pair state (PEPS)} is determined by a graph $G = (W,E)$, bond dimensions $D_e$ for $e \in E$, physical dimensions $d_x$ for $x \in W$ and a collection of linear maps
	\begin{align*}
		T_x : \bigotimes_{e = (xy) \in E} \H_{e,x} \to \H_x = \CC^{d_x}.
	\end{align*}
	The PEPS is then given by
	\begin{align*}
		\ket{\Psi} = \left(\bigotimes_{x \in W} T_x \right) \ket{\Phi} \in \bigotimes_{x \in W} \H_x,
	\end{align*}
	where $\ket{\Phi}$ is defined as in \cref{eq:link state}.
\end{dfn}

The state $\ket{\Psi}$ is in general not be normalized.
The normalization can be computed as
\begin{align}\label{eq:normalization}
	\braket{\Psi | \Psi} = \bra{\Phi} \bigotimes_{x \in W} T_x^\dagger T_x \ket{\Phi}.
\end{align}
In our convention, $\ket{\Psi}$ is normalized if all $T_x$ are isometries.

The graph $G$ can in principle be arbitrary, but of special interest are lattices, either with open or periodic boundaries.
For the special case of a one-dimensional lattice, the resulting PEPS is known as a \emph{matrix product state} (MPS).
If the graph is a lattice with periodic boundary conditions, it is of interest to consider states where one uses the same linear map at each site. Such PEPS are referred to as \emph{uniform} (or translation invariant).

The PEPS is called \emph{injective} if the maps $T_x$ are all injective.
We would like to use a quantitative version of this.
Given injective $T_x$, let $\sigma_{1}(T_x) \geq \dots \geq \sigma_{d_x}(T_x)$ denote its singular values.
The condition number of $T_x$ is given by the ratio of the largest and smallest singular values
\begin{align*}
	\kappa(T_x) = \frac{\sigma_1(T_x)}{\sigma_{d_x}(T_x)}.
\end{align*}

\begin{dfn}[$\delta$-injective PEPS]
	Let $\ket{\Psi}$ be a PEPS constructed from linear maps $T_x$ as in \cref{def:peps}.
	We say that $\ket{\Psi}$ is \emph{$\delta$-injective} for $0 < \delta \leq 1$ if the condition number $\kappa(T_x) \leq \delta^{-1}$ for all $x \in W$.
\end{dfn}

Without loss of generality, we take $\norm{T_x}_{\infty} = \sigma_1(T_x) = 1$ (this only changes normalization), in which case the PEPS is $\delta$-injective if the smallest singular value of across $T_x$ is at least $\delta$.
If a PEPS is $1$-injective, it is called \emph{isometric}, since in this case all $T_x$ are isometries.

\subsection{Parent Hamiltonians}
A useful property of injective PEPS is that one can define an associated Hamiltonian which is local with respect to the graph $G$, which is frustration-free, and has the PEPS $\ket{\Psi}$ as its unique ground state.
Specifically, we can define a parent Hamiltonian by
\begin{align}
	H     = \sum_{e \in E} h_e                                                                         \quad \text{ where } \quad
	h_{e} =  (T_{x}^{-1})^\dagger (T_{y}^{-1})^\dagger \left(\id-\proj{\phi_e}\right)T_{x}^{-1}T_{y}^{-1} \quad \text{ for } e = (xy).
	\label{eq:parent-hamiltonian}
\end{align}
It is easy to see that Hamiltonian is such that each term $h_e \geq 0$, and $ h_e \ket{\Psi} = 0$ (so $H$ is frustration-free), and $\ket{\Psi}$ is the unique ground state.

\subsection{Computational tasks}

We now define two computational tasks associated to PEPS contractions.
For these computational tasks, the input consists of the graph $G = (W,E)$ and the local maps $T_x$ for $x \in W$.
The size of the problem is measured in terms of the size of the graph $n = \abs{W}$, the maximal bond dimension $D_e$, and the required precision $\eps$.
In most cases we will assume that the bond dimension is constant and consider primarily two-dimensional lattices.

\begin{dfn}[\normalization]\label{dfn:norm}
	Given a PEPS $\ket{\Psi}$ defined by $G = (W,E)$ and maps $T_x$ for $x \in W$, compute the norm $N = \braket{\Psi | \Psi}$.
\end{dfn}
In applications, one typically is interested in the normalized state $N^{-\frac12} \ket{\Psi}$.
It can be reasonable to require exponential additive precision in $n$ for $\normalization$; if the PEPS is assumed to be $\delta$-injective the norm $N$ can be as small as $\delta^{2n}$.
Alternatively, one can require multiplicative precision.
In many cases, one would like to know expectation values of local observables, with respect to the normalized state.

\begin{dfn}[\nlev, Normalized expectation value of local observables]\label{dfn:nlev}
	Given a PEPS $\ket{\Psi}$ defined by a graph $G = (W,E)$ and maps $T_x$ for $x \in W$, and a Hermitian operator $O_X$ with support on a constant number of vertices $X \subseteq W$, compute
	\begin{align*}
		\langle O_X \rangle = \frac{\bra\Psi O_X \ket\Psi}{\langle \Psi\ket\Psi}.
	\end{align*}
	to constant additive precision $\eps$.
\end{dfn}
We can also consider the \emph{decision} version of this problem, which is to decide whether either $\langle O_X \rangle \leq \frac13$ or $\langle O_X \rangle \geq \frac23$.
This task is important, since it allows one to evaluate the energy (density) of a PEPS for a local Hamiltonian, which in turn is relevant for variational optimization. Additionally, it allows one to determine physical properties of the PEPS (such as the magnetization).

\begin{figure}[t]
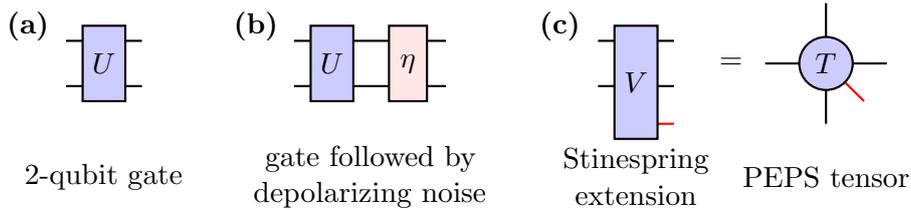

	\centering
	\embeddingFig
	\caption{{Embedding circuits into PEPS~\cite{Malz2025}.}
		\textbf{(a)} A (postselected) circuit is built out of gates (and projections) $U$ acting on a Hilbert space of dimension $q$ (here, $q=2\times2$).
		\textbf{(b)} In the noisy circuit considered here, each gate/projection is followed by depolarizing noise at rate $\eta$, yielding the map $\Phi_\eta$.
		\textbf{(c)} The PEPS tensors $T$ are obtained by considering the Stinespring representation of the map using an ancilla of dimension $q^2$ (red). To obtain $\delta$-injective tensors, we need $\eta=\delta^2/(4(1+3\delta^2))$.
	}
	\label{fig:circuit-embedding}
\end{figure}

\section{Strongly injective PEPS are easy to contract}\label{sec:easy regime}
Our first observation is that estimating local observables to polynomial precision is classically efficient.
It is known that sufficiently injective PEPS have gapped parent Hamiltonians, which implies that the resulting states have exponentially decaying correlations~\cite{Schuch2011}. This can be used to show that to compute local observables to precision $\eps$, one can restrict to a patch of width $\O(\log \eps^{-1})$ of the PEPS around the local observable. This patch can be contracted at cost $\poly(\eps^{-1})$, leading to the conclusion in \cref{cor:easy regime}.
To make this precise, we use the following result from Ref.~\cite{Schuch2011}.

\begin{lem}[Uniform gap for strongly injective PEPS~\cite{Schuch2011}]\label{lem:gap}
	There exists $\delta_{\easy} < 1$ for any lattice, such that $\delta$-injective PEPS with $\delta>\delta_{\easy}$ have a parent Hamiltonian $H = \sum_l h_l$ with local terms with $\norm{h_l}\leq 1$ and with a gap $\Delta > 1/2$ independent of system size.
\end{lem}
Concretely, the argument in Ref.~\cite{Schuch2011} establishes that $\delta_{\easy}$ for the square lattice is at most $0.982$.

Efficient contraction now follows from an argument given in Ref.~\cite{Schwarz2017}. \Cref{cor:easy regime} improves over Ref.~\cite{Schwarz2017}, which obtained a quasipolynomial rather than polynomial algorithm. The reason is that Ref.~\cite{Schwarz2017} bounds the cost of contracting a region of size $\ell \times \ell$ by $\exp(\ell^2)$ rather than $\exp(\ell)$.
Apart from this, the argument is identical; we reproduce it for completeness.

\begin{cor}[\nlev{} is efficient in strongly injective PEPS]\label{cor:easy regime}
	Let $G$ be a lattice in $d$ spatial dimensions. There exists $\delta_{\easy} < 1$, such that for $\delta$-injective PEPS with $\delta>\delta_{\easy}$ on $G$, there exist a classical algorithm to solve $\nlev$ to precision $\eps$ in quasipolynomial time $\exp(\O(\log^{d-1}(\eps^{-1})))$ (independent of the size of the lattice). In particular, for PEPS in two spatial dimensions this is $\poly(\eps^{-1})$. Moreover, these states can be prepared efficiently on a quantum computer.
\end{cor}

\begin{proof}
	We prove this by adapting the proof in Ref.~\cite{Schwarz2017} to our case.
	A key assumption in that paper is that all PEPS obtained by replacing a subset of the tensors by the identity have a parent Hamiltonian with a gap that is larger than some constant $\Delta_*$.
	\Cref{lem:gap} proves this assumption for strongly injective PEPS with $\delta > \delta_{\easy}$ and the rest of this proof repeats the relevant steps from Ref.~\cite{Schwarz2017}.

	To compute an observable $O_X$ on a constant number of sites $X$ with $\norm{O_X}=1$, we first pick a contiguous ring of sites at positions $\{y_i\}_{i=1}^M$ surrounding $X$ such that $\forall i,\,\dist(X,y_i)\geq\ell$ and such that if we replace all tensors at positions $\{y_i\}$ by the identity, the resulting PEPS describes a state that is a tensor product between the patch inside the ring (containing $X$) and the rest of the lattice.
	The resulting patch has width $\O(\ell)$.

	Second, we define a sequence of PEPS $\{\ket{\psi_i}\}_{i=0}^M$ where $\ket{\psi_0} = \ket{\Psi}$ is the state of interest, and $\ket{\psi_i}=\prod_{j=1}^iT_{y_j}^{-1}\ket{\psi_0}$.
	Note that $\ket{\psi_M}$ is the state where all the tensors on sites $y_i$ for $i = 1, \dots, M$ have been replaced by the identity tensor, and where the patch is disentangled from the rest of the lattice.
	Each $\ket{\psi_i}$ is a PEPS with injectivity at least $\delta$ (since the only change compared to the initial PEPS is that some tensors have been replaced by $\id$, which is maximally injective). By \cref{lem:gap}, the corresponding parent Hamiltonians $\{H_i\}$ all have gap at least $1/2$.

	Third, using exponential clustering of gapped ground states, with respect to the observables $O_X$ and $(T_{y_{j+1}}^{-1})^\dagger T_{y_{j+1}}^{-1}$ which are separated by a distance $\ell$ \cite{Hastings2004,Nachtergaele2006}, there exists $\mu > 0$ such that
	\begin{equation}
		\begin{aligned}
			\left| \frac{\bra{\psi_j}O_X\otimes (T_{y_{j+1}}^{-1})^\dagger T_{y_{j+1}}^{-1}\ket{\psi_j}}{\langle \psi_j|\psi_j\rangle}
			- \frac{\bra{\psi_j}O_X\ket{\psi_j}\langle\psi_{j+1}|\psi_{j+1}\rangle}{\langle \psi_j|\psi_j\rangle^2}\right|
			 & \leq \norm{O_X}_{\infty} \norm{(T_{y_{j+1}}^{-1})^\dagger T_{y_{j+1}}^{-1}}_{\infty} e^{-\mu \ell} \\
			 & \leq \delta^{-2}e^{-\mu\ell}.
		\end{aligned}
	\end{equation}
	Equivalently, this relates the expectation values of $O_X$ on $\ket{\psi_j}$ and $\ket{\psi_{j+1}}$ via
	\begin{equation}
		\left| \frac{\bra{\psi_{j+1}}O_X\ket{\psi_{j+1}}}{\langle\psi_{j+1}|\psi_{j+1}\rangle}
		- \frac{\bra{\psi_j}O_X\ket{\psi_j}}{\langle \psi_j|\psi_j\rangle}\right|
		\leq \delta^{-2}\exp(-\mu\ell)\frac{\langle \psi_j|\psi_j\rangle}{\langle\psi_{j+1}|\psi_{j+1}\rangle}
		\leq \delta^{-4}\exp(-\mu\ell),
		\label{eq:eval-bound}
	\end{equation}
	where in the second inequality we have used that the ratio of norms is the inverse of the normalized expectation value of $T_{y_{j+1}}^{-1\dagger}T_{y_{j+1}}^{-1}$
	on $\ket{\psi_j}$, which is lower bounded by $\delta^2$.

	Finally, we apply the bound \cref{eq:eval-bound} iteratively $M$ times and note that $M = \O(d\ell)$ to arrive at
	\begin{equation}
		\left| \frac{\bra{\psi_{M}}O_X\ket{\psi_{M}}}{\langle\psi_{M}|\psi_{M}\rangle}
		- \frac{\bra{\psi_0}O_X\ket{\psi_0}}{\langle \psi_0|\psi_0\rangle}\right|
		= \O\mleft(\ell\delta^{-4}\exp(-\mu\ell)\mright).
	\end{equation}
	By choosing a radius $\ell=c_1\log(\eps^{-1})$ we can obtain a final evaluation error below $\eps$ for some constant $c_1>0$.
	Contracting a region of width $\ell$ in $d$ spatial dimensions can be done at cost $\exp(\ell^{d-1})$, which in particular for $d = 2$ gives a polynomial algorithm in the inverse precision $\eps^{-1}$.
	The fact that the state can be prepared efficiently on a quantum computer is a direct consequence of the adiabatic preparation scheme of Ref.~\cite{Ge2016}, and the fact that by \cref{lem:gap} the spectral gap remains constant while deforming to a state with maximal injectivity (which can be prepared efficiently).
\end{proof}

Note that the above result is for \emph{constant} bond dimension $D$. The scaling of contracting a patch of width $\ell$ with bond dimension $D$ is $D^{\O(\ell)}$, so for $d = 2$ we get a scaling which is $\exp(\log(D) \log(\eps^{-1}))$ (which is polynomial in either $\eps^{-1}$ or $D$ when keeping the other variable fixed).

\section{Weakly injective PEPS are hard to contract}\label{sec:postBQP-proof}
\subsection{Noisy circuits as injective PEPS}\label{sec:noisy circuits}

In this section, we provide a general construction how to embed quantum circuits into injective PEPS following Ref.~\cite{Anshu2024,Malz2025}.

The three operations we require are (i) unitary gate, (ii) reset/SWAP, and (iii) projection. All of them can be described by positive (not necessarily trace preserving) 2-qubit channels
\begin{equation}
	\Phi_v(\rho) = \sum_{\alpha=0}^{m}K_\alpha\rho K_\alpha\dagg,
	\label{eq:Phi}
\end{equation}
where for a unitary gate we have $m=0$ and $K_0=U$.
For a reset we consider the single-qubit channel with Kraus operators $\ket{0}\!\bra{k}$ and for a projection onto $\ket{0}$ we use the single-qubit channel with Kraus operator $\proj{0}$.
These can be made into two-qubit channels by applying an identity channel on the neighbouring qubit (if we want all tensors to have the same format).
Recall that PEPS tensors are maps from the virtual space (the Hilbert space of the four virtual qubits at a vertex) to the physical space, which we take to be 16 dimensional. Thus, $T: \H_2^{\otimes 4}\to\H_{16}$.
We can write the maps in \cref{eq:Phi} as PEPS tensors using the Stinespring dilation
\begin{equation}
	\tilde T_v = \sum_{\alpha=0}^m\ket{\alpha}\!\bra{K_\alpha}.
	\label{eq:stinespring}
\end{equation}
In this way, an arbitrary circuit can be embedded in the virtual space of the PEPS. To make the output of the circuit detectable in the physical state, we need to apply a SWAP between virtual and physical space at the end. In \cref{eq:stinespring}, this is already achieved through the reset gate.
However, since $m<15$ in the operations considered here, the $T_v$ are not injective.
Thus, we instead use the perturbed tensor
\begin{equation}
	T_v = \tilde T_v+\delta\sum_{\alpha=m+1}^{16}\ket\alpha\!\bra{K_\alpha},
	\label{eq:Tv}
\end{equation}
where the remaining $K_\alpha$ are chosen such that $\{\ket{K_\alpha}\}$ is an orthonormal basis.
By construction, the largest singular value of $T_v$ is 1 and the smallest is $\delta$.

The maps that are effectively applied in the virtual space (after tracing the physical leg) now read
\begin{equation}
	\Phi_\delta(\rho)=\Phi(\rho)+\delta^2\sum_{\alpha=m+1}^{16}K_\alpha\rho K_\alpha\dagg=(1-\delta^2)\Phi(\rho)+4\delta^2\tr[\rho]\id/4.
	\label{eq:Phi-delta}
\end{equation}
Note that $\Phi_\delta$ from \cref{eq:Phi-delta} is no longer trace preserving. Indeed, $\tr[\Phi_\delta(\rho)]=\tr[\rho](1+3\delta^2)$. This is an unimportant overall factor that is removed by normalizing the PEPS.
Up to this factor, we have
\begin{equation}
	\Phi_\eta(\rho)= (1-\eta)\Phi(\rho) + \eta\tr[\rho]\id/4,
	\label{eq:noisy-map}
\end{equation}
with effective depolarizing noise
\begin{equation}
	\eta=4\delta^2/(1+3\delta^2),
	\label{eq:eta}
\end{equation}
as illustrated in \cref{fig:circuit-embedding}.

\subsection{Tools from fault-tolerant quantum computation}\label{sec:fault tolerance}

Having established that noisy quantum computations can be encoded into injective PEPS, the next step is to show how we can circumvent these errors to realise nearly error-free quantum computations; in this section we sketch the basic ideas of fault tolerance and state the main results that we need for our purposes.

A protocol for fault tolerance requires the choice of a quantum error-correcting code (QECC). We assume that we can prepare computational basis states in this code, and that we can execute logical gates fault-tolerantly. That is, we can perform 1- and 2-qubit gates on the encoded system via constant-sized circuits, followed by a constant-sized error correction circuit $\QEC$ on each logical qubit that corrects any errors of sufficiently small weight (depending on the choice of QECC). Transforming all of the gates in a circuit $\circuit$ in this way will yield a larger logical circuit, denoted $\enc[\mathcal{C}]$. The constant increase in the size of $\enc[\circuit]$ compared to $\circuit$ will introduce more opportunities for errors to occur, however if the physical noise rate is below some constant threshold $\delta_{\hard}$, then the overall effect will be to reduce the logical error rate.

With such a procedure in hand, the next step is to concatenate the QECC to iteratively reduce the logical error rate further \cite{Knill1996}. That is, we repeat the encoding procedure described above, but now starting from the already encoded circuit $\enc[\circuit]$. Applying this procedure $k$ times yields circuit $\enc^k[\circuit]$ (also let $\QEC_k = \enc^{k-1}[\QEC]$ denote the error correction circuit at the $k$-th level). The size of these circuits is exponential in $k$, but fortunately (provided the noise rate is below $\delta_{\hard}$) the logical error rate decreases doubly-exponentially in $k$ (i.e. as $\exp(-\exp(\Omega(k)))$. In this way, for modest values of $k$ one can achieve a circuit with a vanishingly small logical noise rate \cite{Aharonov1999,Knill1998}. Note that the injective PEPS formalism requires that the fault-tolerant circuits have one-dimensional geometry, and contain no noiseless classical wires or mid-circuit measurements. These restrictions are not prohibitive to obtain fault tolerance, though in practice they will considerably lower the constant error threshold required and may place restrictions on the choice of initial QECC \cite{Aharonov1999,Gottesman2000,Svore2005}.

\begin{thm}[Threshold theorem \cite{Aharonov1999}]\label{lem:threshold theorem}
	Let $\circuit$ be a circuit with an $n$-qubit output, of size $\poly(n)$, and let $k\geq 1$. Then, provided that $\eta<\eta_{\thresh}$, the probability of a logical error in the noisy encoded circuit $\big[\enc^k[\circuit] \big]_\eta$ is bounded by $\poly(n)\exp[-\exp(\Omega(k))]$.
\end{thm}

For post-selection and local measurement readout, we will need to decode from the concatenated QECC. Letting $D$ denote the decoding circuit for a single logical qubit with one layer of quantum error correction, the fault-tolerant decoding circuit for a logical qubit encoded at the $k$th level is $D_k\coloneqq D\circ \enc[D]\circ\dots\circ \enc^{k-1}[D]$. Under a noise rate $\eta < \eta_{\thresh}$, the probability of a logical fault entering at $\enc^{j}[D]$ is doubly exponentially small in $j$, since the $j$ layers of concatenation suppress the error rate doubly exponentially whilst the circuit size only increases exponentially in $j$. Hence the probability that a logical fault occurs anywhere in the circuit $D_k$ can be bounded by $\O(\delta)$, independently of $k$. When decoding multiple qubits, since our noise model consists of local depolarising noise and qubit decoding circuits are independent of one another (in the tensor network, they can be separated by completely depolarising tensors---i.e.~an identity map from the virtual to physical space---which allows no correlations to pass through), the decoding errors affecting different qubits are independent of one another.

\begin{lem}[Concatenated decoding of errors]\label{lem:decoding error}
	Let $\circuit$, $k$, and $\eta$ be as in \cref{lem:threshold theorem}. Then for the noisy encoded circuit followed by decoding $\big[D_k^{\otimes n} \circ \enc^k[\circuit] \big]_\eta$, the probability that a decoding error occurs on a given qubit is $\O(\eta)$, independently of the errors on other logical qubits.
\end{lem}

\subsection{Fault-tolerant postselected quantum computation}

In this section we explain how to establish a fault-tolerant threshold for postselected quantum computation under constant independent depolarising noise, which is a key technical step in our hardness proof for weakly injective PEPS in light of \cref{sec:noisy circuits}. This may be of independent interest; the question of whether postselected quantum computations can be made noise resilient has previously been raised in the context of complexity theory \cite{aaronson2004quantum,Aaronson2005,abrams1998nonlinear}, and has applications to computational models involving closed timelike curves \cite{bennett2009can,lloyd2011closed,leung2013computational} which we elaborate on in \cref{sec:ctc}. A closely related work is \cite{Fujii2016}, which studies fault-tolerance for a model with noisy circuits but exact postselection, motivated by random circuit sampling problems.

We begin by formalising our definition of postselected computation by recapping the definition of the complexity class $\postBQP$. This class was introduced by Ref.~\cite{Aaronson2005} and shown to equal $\PP$.
\begin{dfn}[$\postBQP$]\label{def:postbqp}
	A decision problem is in $\postBQP$ if it can be solved with a family of uniform polynomial-sized quantum circuits with post-selection. That is, the unitary circuits $U$ have two distinguished output qubits such that
	\begin{itemize}
		\item on all valid inputs, the first qubit has a nonzero probability of being measured in the $\ket{0}$ state, and
		\item after postselecting onto this outcome, the second output qubit will be measured in the $\ket{1}$ state with probability $\leq 1/3$ on all rejecting inputs, and the $\ket{1}$ state with probability $\geq 2/3$ on all accepting inputs.
	\end{itemize}
\end{dfn}

This ideal construction is sketched in \cref{fig:postbqp circuits}(a), where the circuit $U$ is formally assumed to encode the input to the decision problem and takes as input $n$ logical qubits in the $\ket{0}$ state. We will show how such a computation can be made robust, provided that noise is below some constant threshold $\eta < \eta_{\thresh}$. We first explicitly define our noise model.

\begin{dfn}[Noisy postselected quantum circuit]
    We consider postselected quantum circuits under the model of noise given by \cref{eq:noisy-map} --- that is, every unitary gate is subjected to depolarising noise at rate $\eta > 0$, whilst postselection gates $\rho\mapsto \bra{0} \rho \ket{0}$ are replaced by noisy postselection
    \begin{align}
        \rho \mapsto (1-\eta) \bra{0} \rho \ket{0} + \eta \tr[\rho]\ .
    \end{align}
\end{dfn}

\begin{lem}[Fault-tolerant threshold for $\postBQP$]\label{lem:postbqp threshold}
    Let $\mathcal{C}$ be a quantum circuit on $n$ qubits and with $\poly(n)$ gates with one postselection register and one designated output register, as in \cref{def:postbqp} (also see \cref{fig:postbqp circuits}(a)). Then, provided that $\eta < \eta_{\thresh}$, there is a corresponding circuit $\tilde{\mathcal{C}}$ on $\poly(n)$ qubits with $\O(n)$ postselection registers and one designated output register, such that the output of $[\tilde{\mathcal{C}}]_\eta$ (including noisy postselection) is within trace distance $\O(\eta)$ of the output of $\mathcal{C}$, independently of $n$. The conclusion also holds when the circuit $\tilde{\mathcal{C}}$ is restricted to 1-dimensional geometry without noiseless classical wires or mid-circuit measurements.
\end{lem}

\begin{proof}
    \begin{figure}[tb]
        \centering
        \scalebox{.78}{
            \begin{tikzpicture}[thick]
                \centering
                \node at (1.1,4.5) {\textbf{(a) Postselected quantum computation}};
                \node at (12,4.5) {\textbf{(b) Fault-tolerant postselected quantum computation}};
                \unencodedCircuit
                \begin{scope}[shift={(8,-5.5)}]
                    \encodedCircuit
                \end{scope}
            \end{tikzpicture}
        }
        \caption{(a) Ideal postselected quantum computation. A circuit consisting of $\poly(n)$ gates acts on $n$ qubits. The first output qubit is postselected on the $\ket{0}$ state, and the second output qubit is measured in the computational basis. (b) Fault-tolerant circuit for noisy postselected computation. $n$ logical qubits are encoded via $k$-fold concatenation of a quantum error-correcting code, and given as input to a fault-tolerant version of the unitary $U$. Moreover, $m-1$ additional qubits are also encoded and the postselection register is ``copied'' onto them using logical $\CNOT$ gates. The resulting $m$ postselection qubits are the noisily decoded and projected onto the $\ket{0}$ state. Meanwhile, the output qubit is decoded and measured.
        }
        \label{fig:postbqp circuits}
    \end{figure}
    Let $U$ be the unitary on $n$ qubits corresponding to the circuit $\mathcal{C}$ as above. The first output qubit is designated as the postselection register, hence we write the circuit output as
    \begin{align*}
        U \ket{0}^{\otimes n} = \alpha_0 \ket{0} \otimes \ket{\psi_0} + \alpha_1 \ket{1} \otimes \ket{\psi_1},
    \end{align*}
    where $\ket{\psi_0}$ and $\ket{\psi_1}$ are normalised states. Crucially, we may assume that $\alpha_0 = 2^{-\O(n)}$~\cite{Aaronson2005}, which below will set the scale for errors in the fault-tolerant computation.
    The postselected state is $\ket{\psi_0}$, upon which we consider the expectation of an observable $O$, $\|O\|=1$ acting on the first qubit,
    \begin{align*}
        x = \braket{\psi_0 | (O \otimes \id ) | \psi_0}.
    \end{align*}
    Our aim is to design a fault-tolerant circuit such that the measurement of $O$ on its output qubit is within $\O(\eta)$ of $x$.
    A problem that arises in the noisy setting is that the available projection gates are not perfect and can suppress the unwanted amplitude only by a constant factor $\eta$ (see \cref{eq:noisy-map}).
    To be able to project down to a level set by the amplitude $\alpha_0=\exp(-\O(n))$, we classically copy the data on the postselection register to $m-1$ additional qubits ($m=\O(n)$) by applying $\CNOT$ gates after the circuit $U$.
    This yields a circuit $\tilde U$ with ideal output
    \begin{equation}
        \ket{\Psi}\coloneqq \alpha_0 \ket{0}^{\otimes m} \otimes \ket{\psi_0} + \alpha_1 \ket{1}^{\otimes m}\otimes \ket{\psi_1}.
        \label{eq:Psi}
    \end{equation}

    Next, we fix $k \geq 1$, and consider the encoded version of $\tilde U$ under a $k$-fold concatenated QECC.
    Following the conclusions of \cref{sec:fault tolerance}, if the noise is below threshold, $\eta < \eta_{\thresh}$, the output of this circuit is correct with a probability $p_{\mathrm{good}}=1-p_{\mathrm{bad}}$, where \cref{lem:threshold theorem} guarantees that $p_{\bad} \leq \poly(n) \exp[-\exp(\Omega(k))]$.
    Thus, we can write the output state after (i) running $\tilde U$ fault tolerantly, (ii) decoding the output qubit and the $m$ postselection qubits using the decoding channel $D_k$ and (iii) tracing the $n-1$ extraneous qubits as (see also the illustration \cref{fig:postbqp circuits}(b))
    \begin{equation}\label{eq:good vs bad}
        \rho_{\out} = p_{\good} \rho_{\mathrm{good}}+ p_{\bad} \rho_{\bad},\qquad
        \rho_{\mathrm{good}} = \left[\V^{\otimes m}\otimes \V \otimes \tr_{n-1} \right] (\proj{\Psi}),
    \end{equation}
    where the channel $\V$ corresponds to the decoding errors on a single qubits, which can be constructed to be all the same and independent (\cref{lem:decoding error}).
    Let $c(\eta)=\O(\eta)$ (independent of $n,m,\tilde U$), such that we can decompose the single-qubit error channels as $\V=(1-c(\eta))\mathcal I+c(\eta)\Delta$, where $\mathcal I$ is the identity channel, and $\Delta$ is some other channel.
    
    We now consider the effect of projecting the $m$ projection qubits using the noisy projection channel $\Phi_\eta[\cdot]$ from \cref{eq:noisy-map}, followed by tracing them out.
    This process is characterized through
    \begin{subequations}
        \begin{align}
            \tr[\Phi_\eta\circ\V(\proj0)] &=(1-\eta)[(1-c(\eta))+c(\eta)\bra0\Delta(\proj0)\ket0]+\eta \eqqcolon 1-\eps_0,\\
            \tr[\Phi_\eta\circ\V(\proj1)] &=(1-\eta)c(\eta)\bra0\Delta(\proj1)\ket0+\eta \eqqcolon \eps_1,\\
            \tr[\Phi_\eta\circ\V(\ket0\!\bra1)] &= \tr[\Phi_\eta\circ\V(\ket1\!\bra0)] = (1-\eta)c(\eta)\bra0\Delta(\ket0\!\bra1)\ket0 \eqqcolon \eps_{01}.
        \end{align}
        \label{eq:projection-errors}
    \end{subequations}
	using in the last line that $\tr[\Delta(\ket0\!\bra1)] =\tr[\Delta(\ket1\!\bra0)]  = 0$.
    In the following, we pick $\eta<\eta_0$ small enough such that the errors obey the inequality $|\eps_0|,|\eps_1|,|\eps_{01}|< 1/4$.
    This allows us to express the state of the output qubit (\cref{eq:good vs bad}) after noisy postselection on all $m$ ancillary qubits as
    \begin{equation}\label{eq:reduced state output qubit}
        \begin{aligned}
              \rho_{\mathrm{result}}&=\left[ \left( \tr\circ \Phi_\eta \right)^{\otimes m} \otimes \mathcal I \right] \rho_{\mathrm{out}} \\
             &= p_{\mathrm{good}}\left( |\alpha_0|^2 (1-\eps_0)^m\V(\rho_0) + (\eps_1^m+2\mathrm{Re}(\eps_{01}^m))\rho_1 \right)
            +p_{\mathrm{bad}}\rho_{\mathrm{bad}}^{\mathrm{proj}},
        \end{aligned}
    \end{equation}
    where $\tr[\rho_{\mathrm{bad}}^{\mathrm{proj}}]\leq1$, $\tr[\rho_1]\leq1$, and $\rho_0=(\mathcal I \otimes\tr_{n-1})(\proj{\psi_0})$ is the (normalized) state of the output qubit for a noiseless circuit with perfect projection.
    Note that $\rho_{\mathrm{result}}$ is subnormalized, and we write its trace as $\tr[\rho_{\mathrm{result}}]=\N+\N_\eps$, where
    \begin{subequations}
		\begin{align}
			\N&\coloneqq p_{\mathrm{good}}|\alpha_0|^2 (1-\eps_0)^m,\\
			\N_\eps&\coloneqq \tr[p_\mathrm{good}(\eps_1^m + 2\mathrm{Re}(\eps_{01}^m)\rho_1+p_\mathrm{bad}\rho_\mathrm{bad}^\mathrm{proj}]
			\leq4^{-(m-1)}+\poly(n)\exp[-\exp(\Omega(k))].
		\end{align}
		\label{eq:trace-rho}
    \end{subequations}
	Since $\N=2^{-\O(n)}(3/4)^m$, we have that
	\begin{equation}
		\N_\eps/\N=2^{\O(n)}\left\{3^{-m}+\poly(n)\exp[-\exp(\Omega(k))]\right\},
		\label{eq:N-eps}
	\end{equation}
	which can be made small by choosing $k=\Omega(\log n)$ and $m=\Omega(n)$.

    Let us now compute a normalized expectation value on $\rho_{\mathrm{result}}$
	\begin{equation}
    \begin{aligned}
		\frac{\tr[O\rho_{\mathrm{result}}]}{\N+\N_\eps}&=\tr[O \mathcal V(\rho_0)] + \O(\N_\eps / \N)\\
        &= [1-c(\eta)]\tr[O\rho_0] + c(\eta)\tr[O\Delta(\rho_0)]+\O(\N_\eps/\N).
		\label{eq:expect}
    \end{aligned}
	\end{equation}
    With the above choices of $k$ and $m$, the latter term is exponentially small in $n$, and the right-hand side reduces to $\tr[O\rho_0] + \O(\eta)$. Since $\|O\|=1$ was arbitrary, we can conclude that
    \begin{align}
        \bigg\| \frac{\rho_{\mathrm{result}}}{N + N_\epsilon} - \rho_0 \bigg\|_1 = \O(\eta)\ ,
    \end{align}
    as required.
\end{proof}

\subsection{Local expectation value is \postBQP-complete}

\begin{thm}[Hardness of injective PEPS]\label{thm:hardness injective peps}
	There exists a constant $\delta_{\hard} > 0$ such that the decision version of $\nlev$ is $\postBQP$-complete in $\delta$-injective PEPS for $\delta < \delta_{\hard}$.
\end{thm}

This theorem is now a simple consequence of \cref{lem:postbqp threshold}, using the embedding of noisy postselected circuits into injective PEPS given by \cref{sec:noisy circuits}.

\begin{proof}
    The containment of $\nlev$ within $\postBQP$ is an immediate consequence of the $\postBQP$-completeness of contracting general PEPS \cite{Schuch2007}. 
    For hardness, consider a general $\postBQP$ problem instance corresponding to a unitary circuit $U$ as in \cref{def:postbqp}. Using \cref{lem:postbqp threshold}, provided that $\eta < \eta_{\thresh}$, we can construct a corresponding circuit whose postselected output, even under the presence of noise, is within $\O(\eta)$ of the ideal output. We set $\eta_1 > 0$ to be a constant such that the output error is less than $1/3$, so that the promise gap of the $\postBQP$ problem is not closed. Thus, the circuit can be embedded into a $\delta$-injective PEPS for $\delta > 0$ as in \cref{sec:noisy circuits} such that the $\postBQP$ problem instance is decided by evaluating a normalised local observable. Quantitatively, defining $\eta_{\hard} := \min\{\eta_{\thresh},\eta_1\}$, we obtain (by inverting \cref{eq:eta}) the threshold 
    \begin{align}
        \delta_{\hard} := \sqrt{\frac{\eta_{\hard}}{4-3\eta_{\hard}}} > 0\ .
    \end{align}
    
\end{proof}

\subsection{Postselection from closed timelike curves}\label{sec:ctc}

\begin{figure}
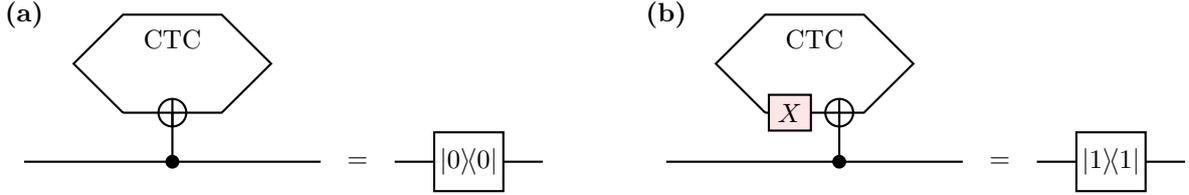

    \centering
	\CTCFig
    \caption{(a) Circuit for a postselection gate onto the $\ket{0}$ state using a closed timelike curve (CTC). (b) A Pauli $X$ error on the CTC corresponds to incorrectly postselecting onto $\ket{1}$; $Y$ and $Z$ errors can be ignored, as these result in zero amplitude over the CTC.}
    \label{fig:noisy ctc}
\end{figure}

In this section we mention a different application of \cref{lem:postbqp threshold}, motivated by the study of closed timelike curves (CTCs) \cite{bennett2009can,deutsch1991quantum,thorne1992closed}. A projected, or postselected, CTC \cite{lloyd2011closed} consists of two quantum systems (one of which is interpreted as travelling backwards in time) which are initialised in a maximally entangled state. The forwards-travelling system undergoes unitary evolution, possibly interacting with external systems, before the two halves of the CTC are projected back onto the maximally entangled state. This model is powerful enough to enable the solution of $\postBQP$ problems~\cite{lloyd2011closed}; \cref{fig:noisy ctc}(a) shows how a unitary $\CNOT$ gate acting on an ancillary closed timelike curve can realise postselection onto the $\ket{0}$ state of a qubit. 

It is natural to ask what happens when the CTC model is subjected to noise \cite{leung2013computational} --- that is, when the unitary gates interacting with the CTC are performed noisily. As illustrated in \cref{fig:noisy ctc}(b), a Pauli $X$ error on the CTC before the $\CNOT$ gate causes the circuit to instead postselect onto $\ket{1}$. On the other hand, $Y$ and $Z$ errors annihilate the state. Hence the effect of depolarising noise applied to the CTC yields a noisy postselection channel of the form (up to normalisation)
\begin{align}
    \Phi_{CTC} : \rho \mapsto (1-q) \bra{0} \rho \ket{0}\cdot \proj{0} + q \bra{1}\rho\ket{1} \cdot \proj{1}.
\end{align} 
After tracing out this postselection qubit, this results in a CP map $\tr\circ \Phi_{CTC}$ equivalent to the noisy postselection model of \cref{sec:noisy circuits}. Hence by the same argument as \cref{lem:postbqp threshold}, one can construct a version of any $\postBQP$ circuit which is fault-tolerant under this noise (with an overhead of $m=\O(n)$ additional CTCs). This establishes that the computational power of projected CTCs is robust to noise.

We note that it may be possible to obtain an alternative proof of \cref{thm:hardness injective peps} by directly embedding a CTC into the injective PEPS, rather than decoding and postselecting, since the time direction of the embedded circuit is arbitrary and can contain loops. In this approach one should show that postselection via CTCs is not compromised by a large subspace of logically equivalent states, and that direction changes in the circuit can be performed fault-tolerantly; we do not pursue this direction here.

\subsection{Hardness of normalization in the uniform case}\label{sec:hardness uniform}
\Cref{thm:hardness injective peps} applies to non-uniform PEPS. One can also study the uniform case, where each vertex in the graph is assigned the same tensor (and we choose periodic boundary conditions). A way to prove hardness of $\normalization$ in the non-injective case is by encoding Wang tiling problems~\cite{Scarpa2020}. Specifically, consider a set of $t$ tiles with $D$ boundary colors. If we have $i = 1, \dots, t$ and the $i$-th tile has sides labeled by colors $a_i, b_i, c_i. d_i$, a tiling is an assignment of tiles to the vertices of the grid, such that the sides have matching colors.
We let $Z(n_1,n_2)$ denote the number of $n_1 \times n_2$ tilings with periodic boundary conditions. Computing this is a $\sharpp$-hard problem.
To this tiling problem we associate a tensor
\begin{equation}
	T = \sum_{i=1}^t \ket{i}\! \bra{a_i}\bra{b_i}\bra{c_i}\bra{d_i}.
\end{equation}
Then, the norm of the resulting PEPS is $Z(n_1,n_2)$. More details on this construction can be found in Ref.~\cite{Scarpa2020}.

One can use this to show that for tensors with constant injectivity, computing $\normalization$ to exponential (additive) precision remains $\sharpp$-hard. The argument is by polynomial extrapolation, as in Ref.~\cite{Haferkamp2020}. We sketch the proof.
One can relate the tiling problem to a problem with injective tensors by letting $T(\delta) = (1 - \delta) T + \delta I$ where $I$ is the maximally injective identity tensor. This gives a state $\ket{\Psi_{n_1,n_2}(\delta)}$ on the $n_1 \times n_2$ grid, and the PEPS has injectivity $\delta$. For $\delta = 0$ we get back the state counting tilings. The norm of $\ket{\Psi_{n_1,n_2}(\delta)}$ is a polynomial of degree $2n$ in $\delta$. If we have an algorithm computing the norm of the PEPS at $\Omega(n)$ evenly spaced values of $\delta$ for $\delta \in (\frac12, 1)$ to exponential precision, then standard results in approximation theory (see \cite{Haferkamp2020} for details) show that we can compute the norm of the state at $\delta = 0$ to exponential precision by polynomial extrapolation, so the problem is already hard at constant injectivity.

\section*{Acknowledgments}
We thank Ignacio Cirac for insightful discussions and Quynh Nguyen for helpful feedback.
DM acknowledges financial support by the Novo Nordisk Foundation under grant numbers NNF22OC0071934 and NNF20OC0059939.

\bibliographystyle{apsrev4-2-titles}
\bibliography{jabref,bibtemp}
\end{document}